\documentclass[letterpaper, 10 pt, conference]{ieeeconf}  
\IEEEoverridecommandlockouts
\usepackage{amsmath,bm}
\usepackage{color}
\usepackage{multirow}
\usepackage{mathrsfs}
\usepackage{lineno,hyperref}
\usepackage{lipsum}
\usepackage{dsfont}
\usepackage{mathtools}
\usepackage{amssymb}

\newtheorem{remark}{Remark}
\newtheorem{lemma}{Lemma}

\newtheorem{assumption}{Assumption}

\newtheorem{proposition}{Proposition}

\def\begcen{\begin{center}}
\def\endcen{\end{center}}

\newcommand{\col}{\mbox{col}}

\def\call{{\cal L}}

\def\calz{{\cal Z}}

\def\calh{{\cal H}}

\def\bp{{s}}

\def\L2e{{\cal L}_{2e}}

\def\rea{\mathbb{R}}

\def\et{\epsilon_t}

\def\l2{{\mathcal L}_2}
\def\l2e{{\cal L}_{2e}}

\def\rea{\mathbb{R}}

\def\begequarr{\begin{eqnarray}}
\def\endequarr{\end{eqnarray}}
\def\begequarrs{\begin{eqnarray*}}
\def\endequarrs{\end{eqnarray*}}
\def\begarr{\begin{array}}
\def\endarr{\end{array}}
\def\begequ{\begin{equation}}
\def\endequ{\end{equation}}
\def\lab{\label}
\def\begdes{\begin{description}}
\def\enddes{\end{description}}
\def\begenu{\begin{enumerate}}
\def\begite{\begin{itemize}}
\def\endite{\end{itemize}}
\def\endenu{\end{enumerate}}

\def\lef[{\left[\begin{array}}
\def\rig]{\end{array}\right]}

\def\begcen{\begin{center}}
\def\endcen{\end{center}}
\def\begrem{\begin{remark}\rm}
\def\endrem{\end{remark}}

\def\TAC{{\it IEEE Trans. Automatic Control}}

\def\SCL{{\it Systems \& Control Letters}}

\def\CST{{\it IEEE Trans. Control Systems Technology}}

\def\begmat#1{\begin{bmatrix}#1\end{bmatrix}}
\def\begali#1{\begin{align}{#1}\end{align}}
\def\begalis#1{\begin{align*}{#1}\end{align*}}

\usepackage{color}

\title{\LARGE \bf
An Adaptive Observer for Sensorless Control of the Levitated Ball Using Signal Injection
}

\author{Bowen Yi, Romeo Ortega, Houria Siguerdidjane and Weidong Zhang
\thanks{This paper is supported by the National Natural Science Foundation of China (61473183, U1509211), China Scholarship Council, the Government of Russian Federation (074U01), the Ministry of Education and Science of Russian Federation (14.Z50.31.0031). {\em (Corresponding author: W. Zhang.)}}
\thanks{Bowen Yi and Weidong Zhang are with Department of Automation, Shanghai Jiao Tong University, Shanghai 200240, China
        {\tt\small yibowen@ymail.com, wdzhang@sjtu.edu.cn}}%
\thanks{Romeo Ortega and Houria Siguerdidjane are with L2S, CNRS-CentraleSup\'elec-Universit\'e Paris Saclay, Plateau du Moulon, Gif-sur-Yvette 91192, France {\tt\small ortega@lss.supelec.fr, houria.siguerdidjane@centralesupelec.fr}}%
}

\begin{document}

\maketitle
\thispagestyle{empty}
\pagestyle{empty}

\begin{abstract}
In this paper we address the problem of sensorless control of the 1-DOF magnetic levitation system. Assuming that only the current and the voltage are measurable, we design an adaptive state observer using the technique of signal injection. Our main contribution is to propose a  new filter to identify the virtual output generated by the signal injection. It is shown that this filter, designed using the dynamic regressor extension and mixing estimator, outperforms the classical one. Two additional features of the proposed observer are that (i) it does not require the knowledge of the electrical resistance, which is also estimated on-line and (ii) exponential convergence to a tunable residual set is guaranteed {without}  excitation assumptions. The observer is then applied, in a certainty equivalent way, to a full state-feedback control law to obtain the sensorless controller, whose performance is assessed via simulations and experiments.
\end{abstract}

\section{Introduction}
\lab{sec1}
%
Magnetic levitation (MagLev) systems are widely used in industry, \emph{e.g.}, rocket-guiding projects, high speed rail transportation, bearingless motors, vibration isolation, magnetic bearing, bearingless pumps, and microelectromechanical systems, see \cite{HANKIM,schw09book} for recent reviews of Maglev systems applications. The inherent instability and high nonlinearity of MagLev systems, make them a theoretical benchmark in the nonlinear control community, with a prototype example being the simple levitated ball.

Detecting the position of the moving objects in MagLev systems needs highly expensive sensors, which usually have low accuracies. These facts stimulate the research for sensorless (also called self-sensing) control of MagLev systems, which require only the measurement of the electrical coordinates. Several technologically motivated sensorless controller have been reported by the applications community \cite{schw09book}. However, to the best of the authors' knowledge, besides some results based on linearized models, {\em e.g.}, \cite{gluck11cep,Miz96tcst}, there are no model-based designs reported in the control community---even for the widely studied levitated ball. One plausible explanation for this situation is that the dynamic structure of Maglev systems does not fit into the mathematically-oriented structures studied by the observer design community \cite{bes07book,gaut01book}, with the additional difficulty that the system is not uniformly observable.

To overcome the first difficulty, in \cite{bobt17auto} a system-tailored observer, that exploits the particular structure of the MagLev model, and the corresponding (certainy equivalence-based) sensorless controller, were proposed. The design relies on the use of parameter estimation-based observers (PEBO) \cite{rom15scl}, which combined with the dynamic regressor extension and mixing (DREM) parameter estimation technique \cite{stastac17,romauto17}, allow the reconstruction of  the magnetic flux. This is later used, by two suitably designed observers for the mechanical coordinates. The loss of observability problem mentioned above, hampers the application of this scheme in ``underexcited" situations, hence requiring an---{\em a priori} unverifiable---richness assumption.

An alternative to overcome the observability obstacle in general nonlinear systems is explored in  \cite{yiscl18} where, following \cite{combesacc16}, probing high-frequency signals are injected in the control, to generate (so-called) virtual outputs used for the observer design. To detect the virtual output, the filter proposed in \cite{combesacc16}, see also \cite{COMetal}, is  applied for PEBO design to the levitated ball and a two-tank system in \cite{yiscl18}. Although the correct asymptotic behavior of the filter is theoretically guaranteed, a bad transient performance, and strong sensitivity to the tuning---and system---parameters was observed for the levitated ball. In particular, the performance was significantly degraded with variations in the systems resistance, that are unavoidable in a practical situation.

In this paper we propose to replace the aforementioned filter by an adaptive scheme that, besides ensuring a better transient performance, removes the need of knowing the systems resistance. Similarly to \cite{bobt17auto}, the design of the new  filter and the resistance estimator, use the DREM estimator, yielding a gradient descent-like adaptive observer. As usual in DREM \cite{stastac17,romauto17}, a key step is the suitable choice of the operators used for the construction of the extended regressor matrix. A central contribution of the paper is to propose a weighted zero-order-hold (WZOH) operator \cite{MIDGOO}, which combined with a delay operator, generates a suitable scalar regressor that---due to the use of probing signals---verifies the excitation condition required to recover the virtual output. Using the latter, flux, position and velocity of the levitated ball system are, then, easily estimated. It is shown that the estimation errors converge exponentially fast into a tunable residual set, thus ensuring some good robustness properties.

The remainder of the paper is organized as follows. Section \ref{sec2} briefly introduces the model of the levitated ball and formulates its state observer and sensorless control problems. Section \ref{sec3} presents the new robust virtual output estimator. In Sections \ref{sec4} and \ref{sec5}, the adaptive observer is designed and analyzed. Simulations and experimental results are given in Section  \ref{sec6}. The paper is wrapped-up with concluding remarks and future research directions in \ref{sec7}.\\

\textbf{Notation.}
$\et$ is a generic exponentially decaying term with a proper dimension. With the standard abuse of notation, the Laplace transform symbol $s$ is used also to denote the derivative operator ${d \over dt}$. $\mathcal{O}$ is the uniform big O symbol, that is, $f(z,\varepsilon)=\mathcal{O}(\varepsilon)$ if and only if $|f(z,\varepsilon)|\le C\varepsilon$, for a constant $C$ independent of $z$ and $\varepsilon$. For an operator $\calh$ acting on a signal we use the notation $\calh[\cdot](t)$, when clear from the context, the argument $t$ is omitted.

\section{Model and Problem Formulation}
\label{sec2}

The classical model of the unsaturated, levitated ball depicted in Fig. 1 is given as \cite{schw09book}
\begequ
\label{model}
\begin{aligned}
    \dot{\lambda} & = - Ri + u \\
    \dot{q} & = {1 \over m} p \\
    \dot{p} & = {1 \over 2k} \lambda^2 - mg\\
    \lambda &={k \over c- q}i,
\end{aligned}
\endequ
where $\lambda$ is the flux linkage, $i$ the current, $q \in (-\infty, c)$ is the position of the ball, $p$ is the momenta, $u$ is the input voltage, $R>0$ is the  resistance, and $m>0$, $c>0$ and $k>0$ are some constant parameters.\\
\begin{figure}
    \centering
    \includegraphics[width=3cm,height=3cm]{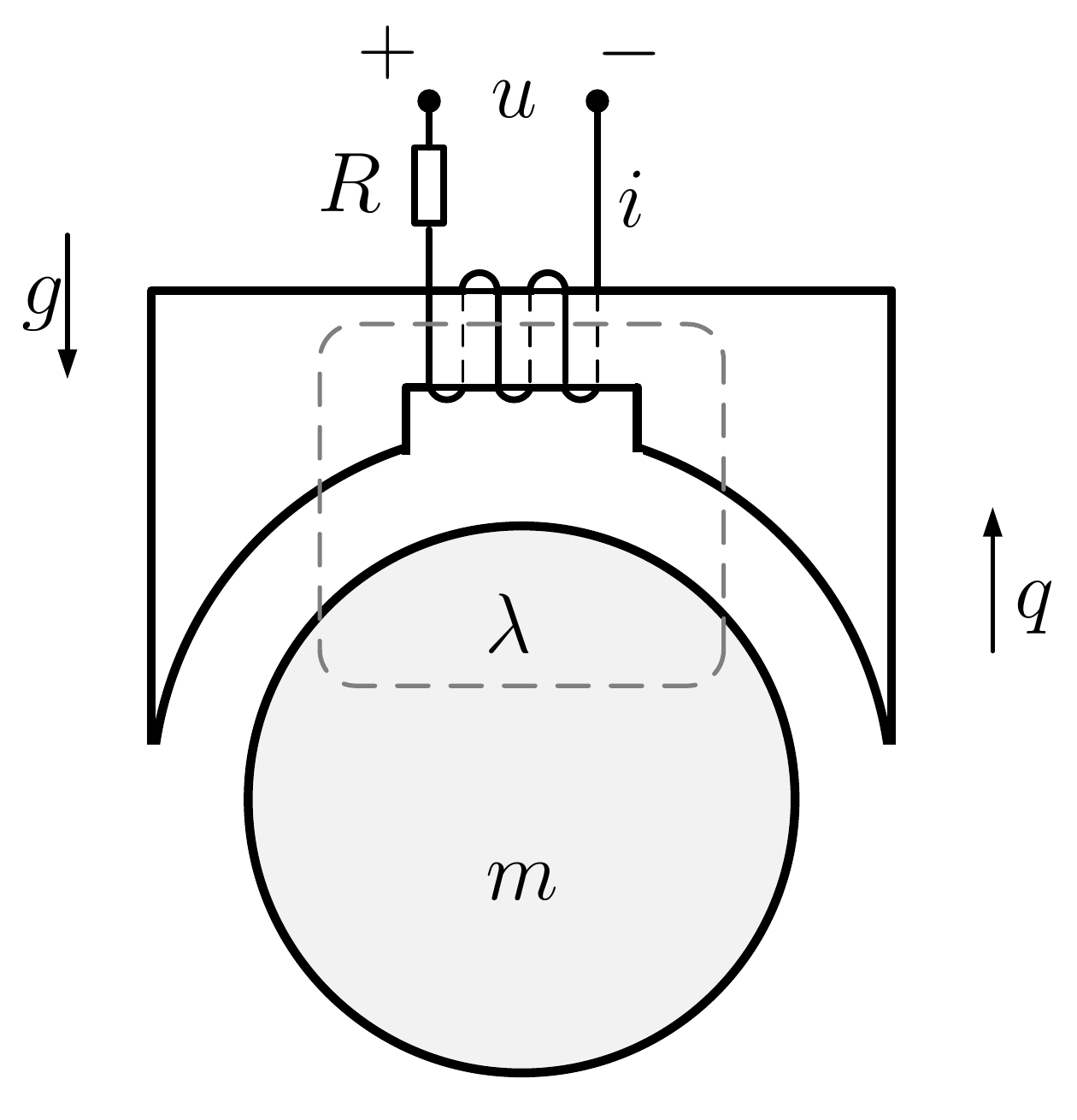}
    \label{fig:maglec}
    \caption{Schematic diagram of the 1-DOF MagLev systems}
\end{figure}
To simplify the notation in the sequel we introduce a change of coordinate for the position and, denoting
$$
x:=  \col(\lambda,q -c,p),
$$
write the system dynamics in the standard form $\dot x=f(x)+gu$ with
\begequ
\lab{fg}
f(x) := \left[\begin{aligned}
                 &{R \over k}x_1 x_2 \\
                & {1\over m}x_3 \\
                 &{1\over 2k}x_1^2 - mg.
               \end{aligned}\right],\;g:=\begmat{1 \\ 0 \\ 0},
\endequ
and define an output
\begequ
\lab{h}
y=h(x):= x_1x_2,
\endequ
which clearly satisfies $y=-ki$.
In this paper we provide a solution to the following.\\

\noindent {\bf Adaptive State Observer Problem.} Consider the dynamics of the levitated ball \eqref{model}, represented in the form $\dot x=f(x)+gu,\;y=h(x)$, with \eqref{fg} and \eqref{h}, the parameters $m$, $c$ and $k$ known, and $R$ {unknown}. Design an adaptive observer
\begequ
\lab{adaobs}
\begin{aligned}
    \dot{\chi} & = F(\chi,u,y)\\
    \hat{x} & = H_(\chi,u,y)
\end{aligned}
\endequ
where $\chi \in \rea^{n_{\chi}}$ is the observer state, such that
\begequ
\label{x_conv}
\limsup_{t\to\infty} \big| \hat{x}(t) - x(t)  \big|  \le \mathcal{O}({\varepsilon})
\endequ
with $\varepsilon>0$ a tunable (small) constant.

As usual in observer design problems we need the following.

\begin{assumption}
\lab{ass1}
Consider the system \eqref{model} with input \eqref{siginj}. The controller signal $u_C$ is such that all the states are {\em bounded}.\footnote{To avoid cluttering the notation we use the generic symbol $\kappa$ to denote a positive constant that upperbounds all signals.}
\end{assumption}

The sensorless controller is obtained applying certainty equivalence to the linear, static-state feedback, asymptotically stabilizing, interconnection and damping assignment passivity-based control (IDA-PBC) reported in \cite{ORTetalcsm}, to ensure
\begequ
\lab{tracla}
 \limsup_{t \to \infty}  |q(t) - q_\star| \le \mathcal{O}({\varepsilon}),
\endequ
where $q_\star$ is the desired position for the levitated ball.

\begrem
\lab{rem1}
We make the important observation that it is possible to show that the system does not satisfy the observability rank condition [Section 1.2.1]\cite{bes07book}, therefore it is not uniformly differentially observable.
\endrem
%
\section{Signal Injection and Virtual Output Filter}
\label{sec3}
%
In order to overcome the lack of observability problem, we follow the signal injection method proposed in \cite{combesacc16}, and further elaborated in \cite{yiscl18}, to generate a new ``virtual" output. As shown in those papers the latter is given by
\begequ
\lab{yv}
y_v=\left({\partial h(x) \over \partial x}\right)^\top g=x_2.
\endequ
To generate $y_v$, we add to the {\em controller output}, denoted $u_C$, a high-frequency sinusoidal signal $s$ to generate the actual input to the system, that is,
\begali{
\nonumber
u & = u_C +  s\\
\label{siginj}
s(t) & = A_0\sin \bigg({2\pi \over \varepsilon} t \bigg),
}
with $\varepsilon>0$. As shown in \cite{combesacc16,yiscl18}, a second-order averaging analysis establishes that, there exists $\varepsilon_\star>0 $ such that, for all $\varepsilon \in (0,\varepsilon_\star]$, we have that the following identity in the interval $[t,t+\mathcal{O}(1/\varepsilon))$,
\begali{
\nonumber
x_1 &= \Bar{x}_1+ \varepsilon S + \mathcal{O}(\varepsilon^2)\\
\nonumber
x_2 &= \Bar{x}_2 + \mathcal{O}(\varepsilon^2)\\
\lab{aveana}
x_3 &=\Bar{x}_3 + \mathcal{O}(\varepsilon^2),
}
where $S$ is the primitive of $s$, that is,
\begequ
\lab{S}
S(t) := - { A_0 \over 2\pi} \cos \bigg({2\pi \over \varepsilon} t\bigg),
\endequ
and the overline denotes the states of the average system, namely,
\begequ
\label{model_ave}
\dot {\bar x}=f(\bar x)+gu_C
\endequ
with $\bar x(0) = x(0)$. Some simple calculations show that
\begequ
\label{rel1}
y = \bar{y} + \varepsilon S y_v + \mathcal{O}(\varepsilon^2),
\endequ
with $\bar y:=\bar x_1 \bar x_2$.

To simplify the notation, and with some obvious abuse of notation, in the sequel we omit the clarification that the averaging analysis only insures the existence of a lower bound on $\varepsilon$ such that \eqref{aveana} holds, and we simply assume that $\varepsilon$ is small enough.

We recall that the problem is to reconstruct $y_v$ out of the measurement of $y$, for which a sliding-window filter is proposed in  \cite{combesacc16}, and also used in \cite{yiscl18}. As discussed in the introduction, the use of this filter generated some serious robustness problem. Consequently, we propose in this paper to replace it by an estimator that, similarly to DREM, implements a gradient descent observer based on a suitable linear regression model.

The first step is then to obtain the linear regression model, making the key observation that, with respect to $S$, the signals $\bar y$ and $y_v$ are slowly time-varying. This motivates us to view \eqref{rel1} as a linear (time-varying) regression perturbed by a small term $\mathcal{O}(\varepsilon^2)$. Whence,  we write \eqref{rel1} as
\begequ
\label{iphithe}
\begin{aligned}
    y & = \phi^\top \theta+ \mathcal{O}(\varepsilon^2)\\
    \theta & := \col(\bar{y}, \varepsilon y_v)\\
    \phi &:= \col(1, S),
\end{aligned}
\endequ
with $y$ the measurable signal, and $\phi$ and $\theta$ playing the roles of known regressor and (slowly time-varying) parameters to be estimated.

We will estimate the parameters $\theta$ using the DREM estimator---we refer the reader to \cite{stastac17} for additional details on DREM. The main idea of DREM is to generate from \eqref{iphithe} a {\em scalar} regression for the ``parameter" of interest, in this case, $\theta_2=\varepsilon y_v$. Although this can be achieved with arbitrary $\call_\infty$-stable, linear operators, the resulting scalar regressor does not necessarily satisfy the excitation conditions required to ensure parameter convergence. It turns out that in our case, the latter is possible, selecting some specific operators as detailed in the lemma below, whose proof is given in the Appendix.

To streamline the presentation of the lemma we define two $\call_\infty$-stable, linear operators, first, the delay operator $\calh_d$, with parameter $d>0$,
\begequ
\lab{hd}
\mathcal{H}_{{d}}[v](t)= v(t-d).
\endequ
Second, the WZOH operator \cite{MIDGOO} $\calz_w$, parameterized by $w>0$, defined as
\begequ
\lab{zw}
\begin{aligned}
\dot{\chi}(t) & = v(t) \\
\calz_w[v](t) & ={1\over w} \big[ \chi(t) - \chi(t- w) \big].
\end{aligned}
\endequ
%
\begin{lemma}
\lab{lem1}
Consider the linear regression \eqref{iphithe} and the operators $\calh_d$ \eqref{hd} and $\calz_w$ \eqref{zw}, with $w=2d$ and $d=n\varepsilon$, for some  $n \in \mathbb{Z}_+$.
Then,
\begequ
\lab{scalinreg}
Y(t) = S(t-d) \theta_2(t-d) + \mathcal{O}(\varepsilon^2),
\endequ
where we defined the measurable signal
\begequ
\lab{Y}
Y: = (\mathcal{H}_d  - \mathcal{Z}_{2d})[y].
\endequ
\end{lemma}

In the proposition below, we propose a simple gradient descent algorithm to identify $\theta_2$ from \eqref{scalinreg}.
%
\begin{proposition}
\label{prop1}
Consider the system \eqref{model} with measurable output \eqref{h} and input \eqref{siginj} verifying Assumption \ref{ass1}. Define the virtual output estimator
\begequ
\label{vir_out_filter}
\begin{aligned}
    \dot{\hat{\theta}}_2 & = \gamma S(t-d) \big[Y(t) - S(t-d) \hat{\theta}_2\big] \\
        \hat{y}_v & = {\hat{\theta}_2 \over \varepsilon },
\end{aligned}
\endequ
where $\gamma>0$ is a tuning gain, $S$ is given in \eqref{S}, $Y$ in \eqref{Y}, $\mathcal{H}_d$ and $\mathcal{Z}_{2d}$ are defined in \eqref{hd} and \eqref{zw}, respectively, with $d=n\varepsilon$, for some  $n \in \mathbb{Z}_+$.   Then,
\begequ
\lab{tilyv}
\lim_{t\to \infty} |\hat{y}_v(t) - y_v(t) | \le \mathcal{O}({\varepsilon}) \quad \text{(exp.)},
\endequ
where $y_v$ is defined in \eqref{yv}.
\end{proposition}

\begin{proof}
Define the error signal
$$
\Tilde{\theta}_2:= \Hat{\theta}_2 - \varepsilon y_v .
$$
Invoking Lemma \ref{lem1}, and replacing \eqref{scalinreg} in \eqref{vir_out_filter} we get
\begin{equation}
\label{theta2dyn}
\dot{\tilde{\theta}}_2
    := - \gamma S^2(t-d)\tilde{\theta}_2 +\mathcal{O}\big(\varepsilon\big).
\end{equation}
From \eqref{S} it is clear that $S(t)$ satisfies
  $$
    \int_{t}^{t+  \varepsilon}S^2(\mu-d) d\mu \ge S_0
  $$
for any $t$ and some $S_0>0$. Therefore, some basic perturbation analysis completes the proof.
\end{proof}

\begin{remark}
The virtual output filters in \cite{combesacc16,yiscl18} compute the estimate of the virtual output by averaging in a moving horizon $[t-n\varepsilon,t]$ the observation error---that is, in fact, an open-loop design. The new filter \eqref{vir_out_filter} provides an alternative \emph{closed-loop} approach, making it robust to unavoidable measurement noise.
\end{remark}

\section{Adaptive Observer Design}
\label{sec4}
%
In this section we  design the adaptive state observer using the estimate of $y_v$ of Proposition \ref{prop1}. To enhance readability it is split into three subsections presenting, respectively, the resistance estimator, the flux observer and the observer for position and velocity.

\subsection{Resistance identification}
\lab{subsec41}
%
Before presenting the resistance estimator we make the observation that, due to the physical constraints $q \in (-\infty,c)$. As expected, we impose this constraint also to its estimate,\footnote{This can easily be done adding a projection operator to the second equation in \eqref{vir_out_filter}, but is omitted for brevity.} hence
\begali{
\nonumber
y_v =q - c\leq \ell < 0 \\
\lab{yvneqzer}
\hat y_v \leq \ell < 0.
}

As expected in adaptive systems design it is necessary to impose an excitation constraint.

\begin{assumption}
\lab{ass2}
Consider the system \eqref{model} with input \eqref{siginj}. The current $i$ is persistently exciting (PE), that is, there exist $T_i>0$ and $\delta_i>0$ such that
\begequ
\lab{ipe}
\int_t^{t+T_i} i^2(\tau)d\tau \geq \delta_i,
\endequ
for all $t \geq 0$.
\end{assumption}
%
\begin{proposition}
\label{prop2}
Consider the system \eqref{model} with input \eqref{siginj} verifying Assumptions \ref{ass1} and \ref{ass2}. Define the resistance estimator as
\begequ
\label{estimator2}
\dot{\Hat{R}} = \gamma_R \phi_R \big(Y_R - \phi_R \Hat{R}\big)
\endequ
with $\gamma_R>0$ a tuning gain, and
\begali{
\nonumber
\dot{v}_1 & = - a v_1 + au \\
\nonumber
\dot{v}_2 & = - a v_2 + a \bigg( {y\over \hat{y}_v}\bigg)\\
\nonumber
\dot \phi_R & = -a \phi_R + {a\over k}y\\
\lab{stareafil}
Y_R & = -v_1 + a{y\over \hat{y}_v} - av_2,
}
with $a>0$ and $\hat y_v$ generated as in Proposition \ref{prop1}.  Then,
$$
\lim_{t\to\infty} \big| \Hat{R}(t) - R\big| \le \mathcal{O}({\varepsilon}) \quad \text{(exp.)}
$$
\end{proposition}

\begin{proof}
From \eqref{h} and \eqref{yv} we have the relationship $x_1 = {y \over y_v}$, which is well defined in view of \eqref{yvneqzer}. Computing the derivative with respect to time yields
$$
{d \over dt}\bigg( {y \over y_v} \bigg) = R{ y \over k}+u .
$$
Applying to the equation above the linear time invariant (LTI) filter ${a \over \bp+ a}$ yields
\begequ
\lab{linregr}
 {a \bp \over \bp+ a}\bigg[ {y \over y_v} \bigg] - { a \over \bp+ a} \big[u\big] = R { {a \over k} \over \bp+ a} \big[y\big]  + \et,
\endequ
where $\et$ is an exponentially decaying term stemming from the filters initial conditions. As shown in \cite{stastac17}, without loss of generality, this term is neglected in the sequel.

Notice now that \eqref{stareafil} is a state realization of the filters
\begequ
\lab{yrphir}
\begin{aligned}
    Y_R & = {a \bp \over \bp+ a}\bigg[ {y \over \hat y_v} \bigg] - { a \over \bp+ a} \big[u\big] \\
    \phi_R & =  {{a \over k}  \over \bp+ a} \big[y\big],
\end{aligned}
\endequ
Motivated by this fact define the auxiliary (ideal) dynamics
$$
\dot{v}^*_2  = - a v^*_2 + a \bigg( {y\over {y}_v}\bigg).
$$
and the signal
$$
Y_R^*:= -v_1 + a{y\over {y}_v} - av^*_2,
$$
and notice that \eqref{linregr} may be written as
$$
Y^*_R=R \phi_R.
$$
Define the signal $\tilde Y_R:=Y_R-Y_R^*$, which upon replacement in \eqref{estimator2}, yields
\begequ
\label{R_error_dyn}
\dot{\Tilde{R}} = - \gamma_R \phi_R^2 \Tilde{R}+ \gamma_R \phi_R \Tilde{Y}_R,
\endequ
where we defined the resistance estimation error $\Tilde{R}:= \Hat{R} - R$. Exponential convergence {\em to zero} of the unperturbed dynamics follows invoking the PE Assumption \ref{ass2} and standard adaptive control arguments \cite{SASBOD}.

To analyse the stability of \eqref{R_error_dyn} define the error $\tilde v_2:=v_2 - v_2^*$, whose dynamics is given as
\begequ
\lab{dotevtwo}
\dot{\tilde v}_{2} = - a \tilde v_2+ a y   \bigg({1 \over y_v} - {1 \over \hat{y}_v}\bigg).
\endequ
Now, we recall \eqref{yvneqzer}, from which we get the following inequality
$$
    \left|{1 \over y_v} - {1 \over \hat{y}_v}\right|  = \left|{ \tilde {y}_v \over y_v\hat{y}_v} \right|
                                                 \le {1 \over \ell^2} |\tilde{y}_v |,
$$
where we defined $\tilde {y}_v:= \hat{y}_v-y_v$. Using  \eqref{tilyv} and the inequality above, and invoking Assumption \ref{ass1} that ensures $|y| \leq \kappa$, from \eqref{dotevtwo} we conclude that
\begequ
\lab{limevtwo}
\lim_{t\to \infty} |\tilde v_2(t) | \le \mathcal{O}({\varepsilon}).
\endequ
The proof is completed noting that a similar property holds for $\tilde Y_R$ and invoking the exponential stability of the unperturbed dynamics.

\end{proof}

\begrem
The assumption that $i$ is PE is not restrictive at all. Actually it is possible to show that this condition can be transferred to the control $u$.\footnote{The details of this proof are omitted for brevity.} Now, since $u$ defined in \eqref{siginj} contains an additive term that is PE, the condition that $u$ is PE will almost always be satisfied.
\endrem
\subsection{Flux observer}
\lab{subsec42}
Before presenting our observer notice that the flux $x_1$ admits the following algebraic observer
\begequ
\label{algobs}
 \Hat{x}_1 = {y \over \hat y_v}.
\endequ
Unfortunately, due to the division operation, such a design is relatively sensitive to measurement noise, making it non-robust.\footnote{In the resistance estimator, although such relationship is used, the LTI filter and the closed-loop gradient descent dynamics reduce the deleterious effects significantly.} To overcome this drawback, we propose below a closed-loop flux observer design.

\begin{proposition}
\label{prop3}
Consider the system \eqref{model} with measurable output \eqref{h} and input \eqref{siginj} verifying Assumptions \ref{ass1} and \ref{ass2}. Define the flux observer
\begequ
\label{flux_obs}
    \dot{\hat{x}}_1  = { \hat{R} \over k} y + u - \gamma_\lambda  (y - \hat y_v \hat x_1).
\endequ
with $\gamma_\lambda>0$ and $\hat y_v$, $\hat x_1$ generated as in Propositions \ref{prop1} and \ref{prop3}, respectively. Assume the current $i$ verifies \eqref{ipe}. Then,
$$
\lim_{t\to\infty} \big| \Hat{x}_1(t) - x_1(t)\big|  \le \mathcal{O}({\varepsilon}) \quad \text{(exp.)}
$$
\end{proposition}

\begin{proof}
From \eqref{flux_obs} we get
\begalis{
    \dot{\hat{x}}_1 & = { {R} \over k} y + u - \gamma_\lambda  (y - \hat y_v \hat x_1)+ { \tilde{R} \over k} y \\
        & = \dot x_1 - \gamma_\lambda  (y - y_v \hat x_1)+ { \tilde{R} \over k} y +\gamma_\lambda \tilde y_v \hat x_1.
}
Define the flux estimation error $\Tilde{x}_1 := \Hat{x}_1-x_1$, whose dynamics is given as
\begalis{
    \dot{\Tilde{x}}_1 & = \gamma_\lambda y_v\Tilde{x}_1 + { \tilde{R} \over k} y +\gamma_\lambda \tilde y_v \hat x_1.
}
From \eqref{yvneqzer} we see that the unperturbed dynamics is exponentially stable and, moreover, from \eqref{flux_obs} we have that $\hat x_1$ is bounded. The proof is completed using these two properties and invoking Propositions \ref{prop1} and \ref{prop2}.
\end{proof}
\subsection{Position and momenta observer}
\lab{subsec43}
Given the definition of the virtual output \eqref{yv}, we trivially obtain an algebraic observer for $x_2$ as follows
\begequ
\label{obs_pos}
    \Hat{x}_2 = \hat y_v.
\endequ

To obtain an observer for the momenta $x_3$ we follow the Kazantzis-Kravaris-Luenberger (KKL) methodology \cite{romauto17} in the proposition below.

\begin{proposition}
\label{prop4}
Consider the system \eqref{model} with measurable output \eqref{h} and input \eqref{siginj} verifying Assumptions \ref{ass1} and \ref{ass2}. Define the momenta observer
\begali{
\nonumber
  \dot{z} & =-{\gamma_p \over m} z + {1 \over 2k} \hat x_1^2 - {\gamma_p^2 \over m}\hat y_v - mg\\
    \Hat{x}_3 & = z+ \gamma_p \hat {y}_v,
\label{kkl}
}
with $\gamma_p>0$ and $\hat y_v$, $\hat x_1$ generated as in Propositions \ref{prop1} and \ref{prop3}, respectively.  Then,
$$
\lim_{t\to\infty} \big| \Hat{x}_3(t) - x_3(t)\big|  \le \mathcal{O}({\varepsilon}) \quad \text{(exp.)}
$$
\end{proposition}

\begin{proof}
Define the signal
\begequ
\lab{t}
T:= x_3 - \gamma_p y_v.
\endequ
Notice that, from the second equation of \eqref{kkl} and \eqref{t}, we get
\begequ
\lab{tilx3}
\hat x_3 -x_3=z-T+\gamma_p \tilde y_v.
\endequ
As usual in KKL observers, the gist of the proof is to show that $z$ ``approaches" $T$. Now, differentiating \eqref{t} we have
$$
\dot T = -{\gamma_p \over m} T - {\gamma_p^2 \over m}y_v + {1 \over 2k}x_1^2 - mg.
$$
Using the first equation of \eqref{kkl} we get
$$
\dot T  - \dot z =- {\gamma_p \over m} (T - z)+-{\gamma_p^2 \over m} \tilde y_v - {1 \over 2k}(\hat x_1 \tilde x_1 -\tilde x_1^2)
$$
From the equation above we conclude that
$$
\lim_{t\to\infty} \big| z(t)  - T(t) | \le \mathcal{O}({\varepsilon}) \quad \text{(exp.)}
$$
The proof is completed replacing \eqref{tilyv} and the limit above in \eqref{tilx3}.
\end{proof}

\begin{remark}
An alternative to the KKL observer above is a standard Luenberger observer
\begalis{
    \dot z_1 & = {1 \over m} z_2 + c_1(\Hat{x}_2 - z_1) \\
    \dot z_2 & = {1 \over 2k} \Hat{x}_1^2 - mg + c_2(\Hat{x}_2 - z_2)\\
    \hat x_3 &= z_2,
}
with $c_1>0$ and $c_2>0$. However, the order of such a design is higher than that of the KKL observer \eqref{kkl}. Moreover, as shown in Section \ref{subsec61} it was observed in simulations that the KKL observer outperforms the Luenberger one and is easier to tune.
\end{remark}
%
\section{Adaptive Observer and Sensorless Controller}
\label{sec5}
%
To solve the adaptive state observation of Section \ref{sec2} we summarize in this section the derivations presented in the previous section. Also, we propose a sensorless controller.

\begin{proposition}
\label{prop5}
Consider the system $\dot x=f(x)+gu,\;y=h(x)$, with \eqref{fg} and \eqref{h}, with input \eqref{siginj} verifying Assumptions \ref{ass1} and \ref{ass2}.  The $7$-order adaptive observer \eqref{adaobs} with mappings
$$
\begin{aligned}
F & =
\left[
\begin{aligned}
&  -\gamma S(t-d) [Y(t) -S(t-d)\chi_{1}] \\
&  - a\chi_{2} + au\\
&  - a\chi_{3} + a \left({y \over \varepsilon \chi_{1}}\right) \\
&  - a \chi_{4} + {a\over k} y \\
&  \gamma_R \chi_{4} \left( \chi_{2} + a{y \over \varepsilon\chi_{1}} - a \chi_{3} - \chi_{4}\chi_{5} \right)\\
&  -{1\over k}y \chi_{5} + u - \gamma_\lambda \left( y -\varepsilon\chi_{1}\chi_{6}\right) \\
&   {\gamma_p \over m} \chi_{7} + {2\over 2k} \chi_{6}^2 - {\gamma_p^2 \over m}\varepsilon\chi_{1} - mg
\end{aligned}
\right] \\
H & =
\begin{bmatrix}
  \chi_{6} \\
  \varepsilon\chi_{1} \\
  \chi_{7} + \gamma_p \varepsilon\chi_{1}
\end{bmatrix},
\end{aligned}
$$
with $Y$ given in \eqref{Y} and $a,\gamma,\gamma_R,\gamma_\lambda,\gamma_p>0$, guarantees \eqref{x_conv}.
\end{proposition}

\begin{proof}
The proof is established identifying
$$
\chi = \col(\hat{\theta}_2, v_1, v_2, \phi_R, \hat{R}, \hat{x}_1, z),
$$
and invoking the results of Propositions \ref{prop1}-\ref{prop4}.
\end{proof}

We are in position to give the sensorless control law, which is a certainty equivalence version of  the IDA-PBC given in \cite{ORTetalcsm}. Namely
\begali{
\nonumber
u_C & = -{1 \over k}\chi_5 y - K_p \bigg( {1\over \alpha} (\chi_6 -\lambda_{\star})  + (\varepsilon\chi_1 + c- q_{\star}) \bigg) \\
\lab{uc}
 &  - \bigg( {\alpha \over m} + K_p \bigg)  (\chi_{7}  + \gamma_p \varepsilon\chi_{1})
}
where $\lambda_\star$ and $q_\star$ are the desired values for $\lambda$ and $q$, respectively, and $K_p,\alpha>0$ are some tuning constants.
%
\section{Simulations and Experiments}
\lab{sec6}
%
In this section, the performance of the novel observer is validated via computer simulations and experiments. All simulations are conducted by Matlab/Simulink. The parameters used in the simulation and the control of the experimental rig are in Table \ref{tab:2}. The new design is compared, via simulations, with the one in \cite{yiscl18}. In both simulations and experiments, the desired equilibrium is $(\sqrt{2kmg}, q_{\star}, 0)$, with $q_{\star}$ taken as a pulse train, and with the initial states $(\sqrt{2kmg}, 0, 0)$.

\begin{table}[h]
\centering
\caption{Parameters of MagLev systems: Simulation (First Column) and Experiments (Second Column)}
\label{tab:2}
\renewcommand\arraystretch{1.6}
\begin{tabular}{l|r|r}
\hline\hline
 Ball mass [kg]&  0.0844&  0.0844  \\
 Gravitational acceleration [$\text{m/s}^2$] &9.81& 9.81\\
 Resistance [$\Omega$] &2.52 & 10.615 \\
 Position ($c$) [m] & 0.005& 0.0079\\
 Inductance constant ($k$) [$\mu$H$\cdot$m]& 6404.2 & 49950\\
 \hline\hline
\end{tabular}
\end{table}

\subsection{Performance of the observer}
\label{subsec61}
For a fair comparison with the observer design in \cite{yiscl18}, simulations are run with the state-feedback version of the controller \eqref{uc}, whose parameters are set as $K_p = 200.7, \alpha = 33.4$. To make simulations more realistic, we add measurement noise in the current $i$, which is generated with the ``Uniform Random Number'' block in Matlab/Simulink, within $[-0.003,0.003]$A.

The parameters in the proposed observer are selected as $A_0=1,\varepsilon=1/300, d=10\varepsilon, a=500, \gamma=3.89\times 10^3,  \gamma_R=500, \gamma_\lambda=8000, \gamma_p= 30$. The parameters of the design in \cite{yiscl18} are selected as $n=5,\varepsilon=1/300, \alpha=0.01,\gamma=50$.

Simulation results in Matlab/Simulink are shown in Figs. \ref{fig:virtual}-\ref{fig:state}, where $\hat{(\cdot)}^*$ denotes the results from the filter in \cite{yiscl18}. To compare the two momenta observers proposed in Subsection \ref{subsec63}, we have included that estimate $\hat{p}_L$, computed by the Luenberger observer \eqref{kkl}. As expected, the new design is less sensitive to measurement noise due to its closed-loop structure, and also, the steady-state observation are of the accuracy $\mathcal{O}(\varepsilon)$. Besides,  the KKL observer outperforms the Luenberger one.

\begin{figure}[]
    \centering
\includegraphics[width=4.2cm,height=3cm]{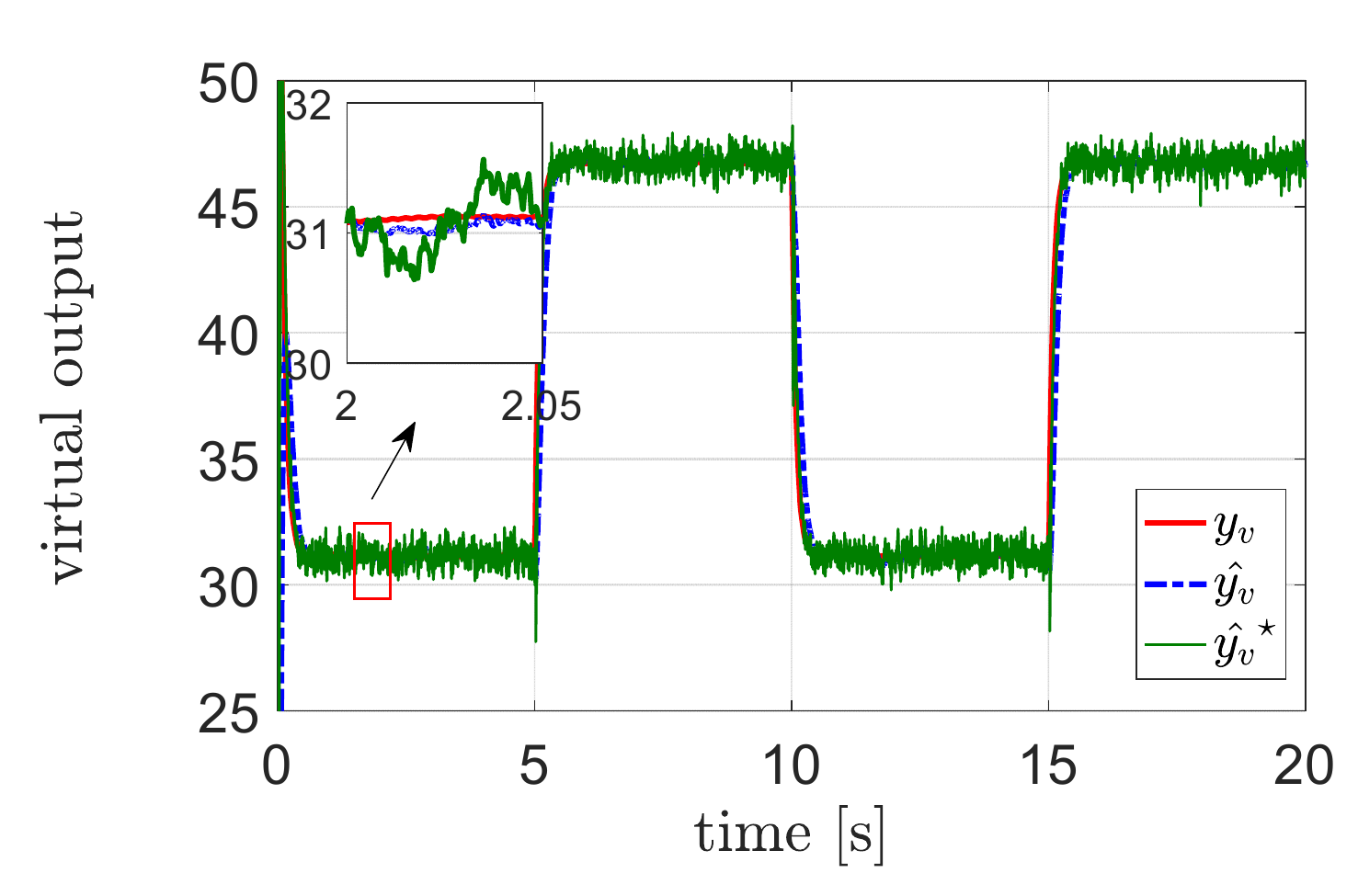}
\includegraphics[width=4.2cm,height=3cm]{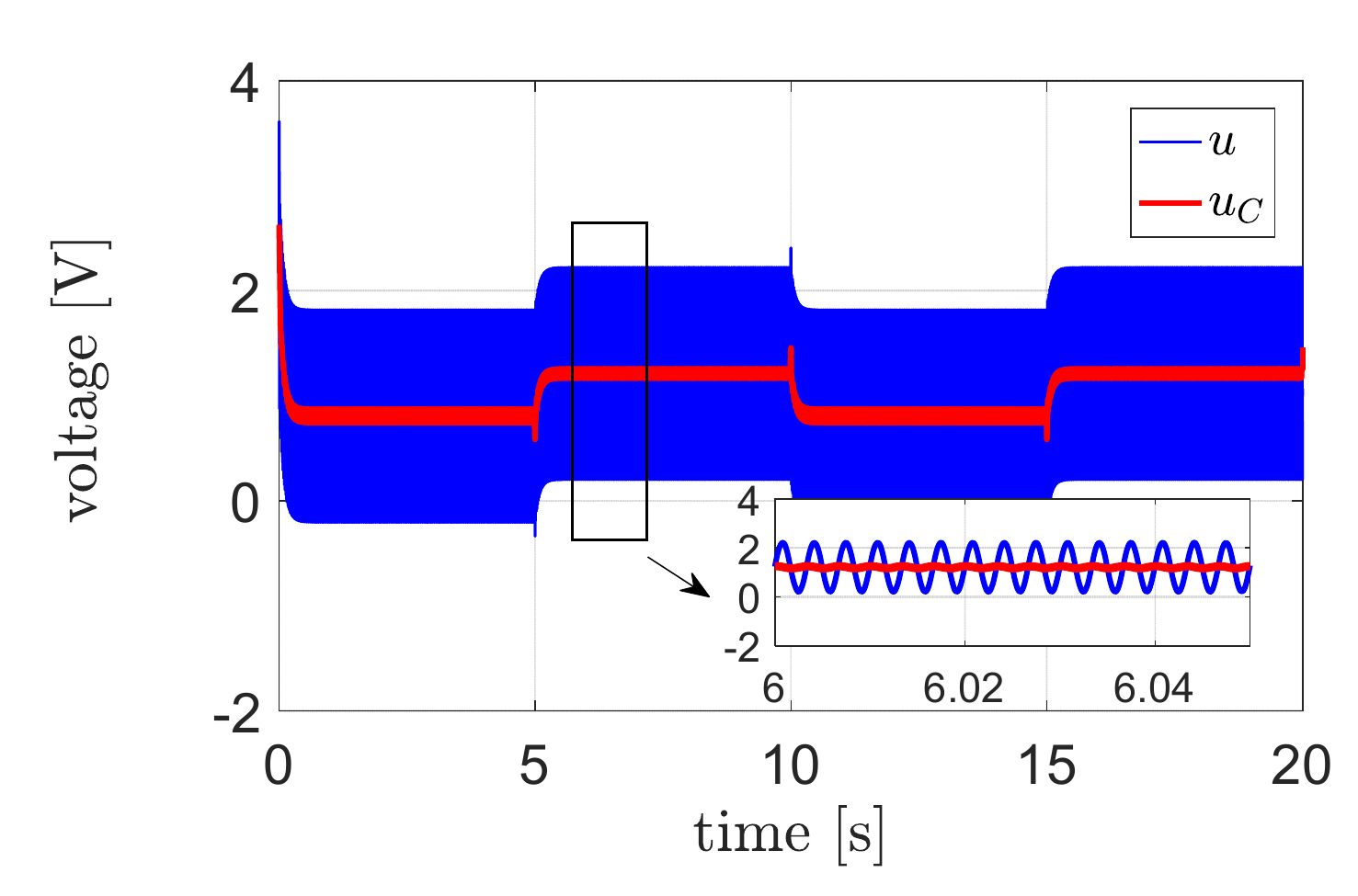}
    \caption{Virtual output estimation and input (simulation)}
    \label{fig:virtual}
\end{figure}
\begin{figure}[]
    \centering
\includegraphics[width=4.2cm,height=3cm]{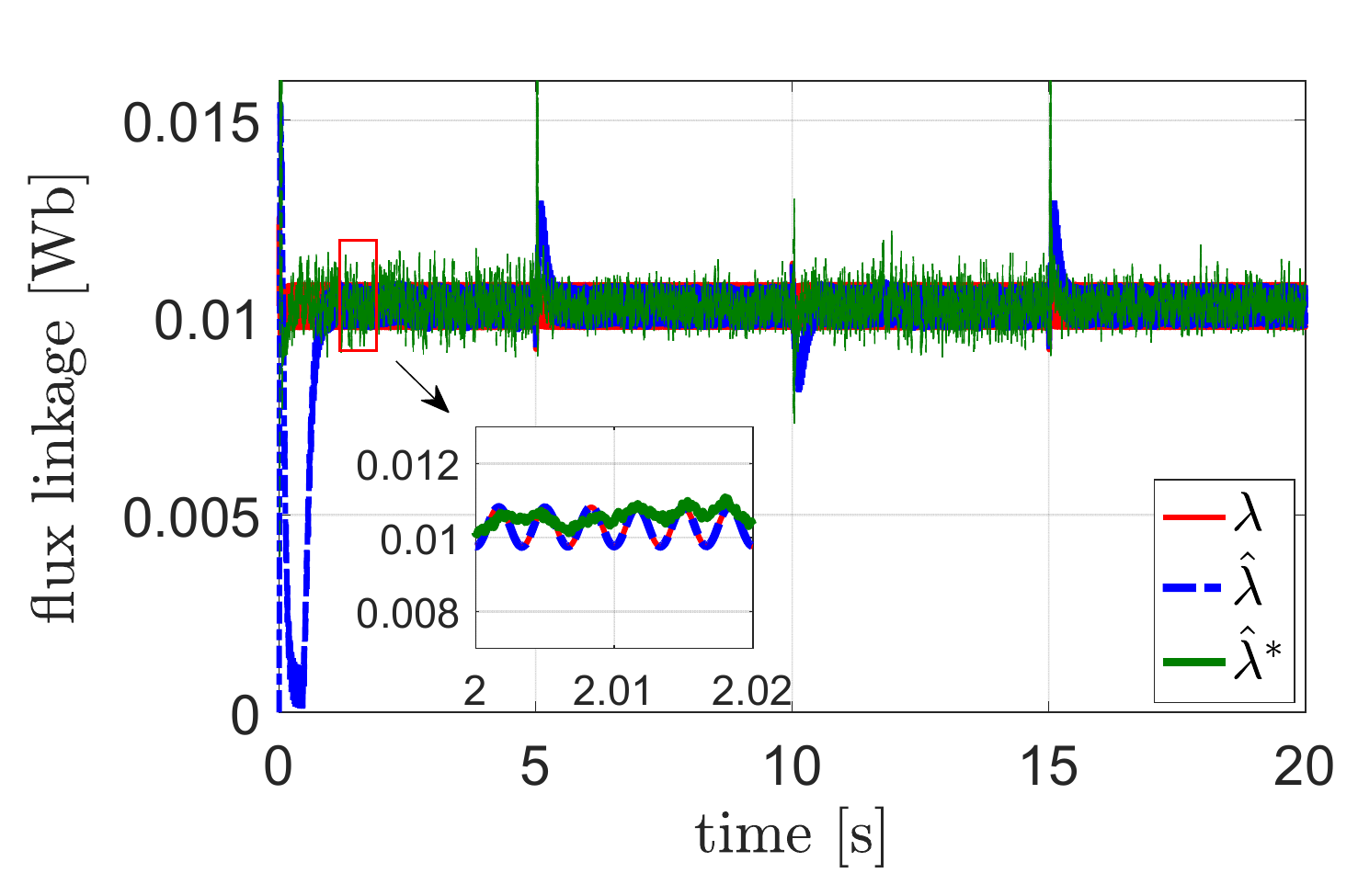}
\includegraphics[width=4.2cm,height=3cm]{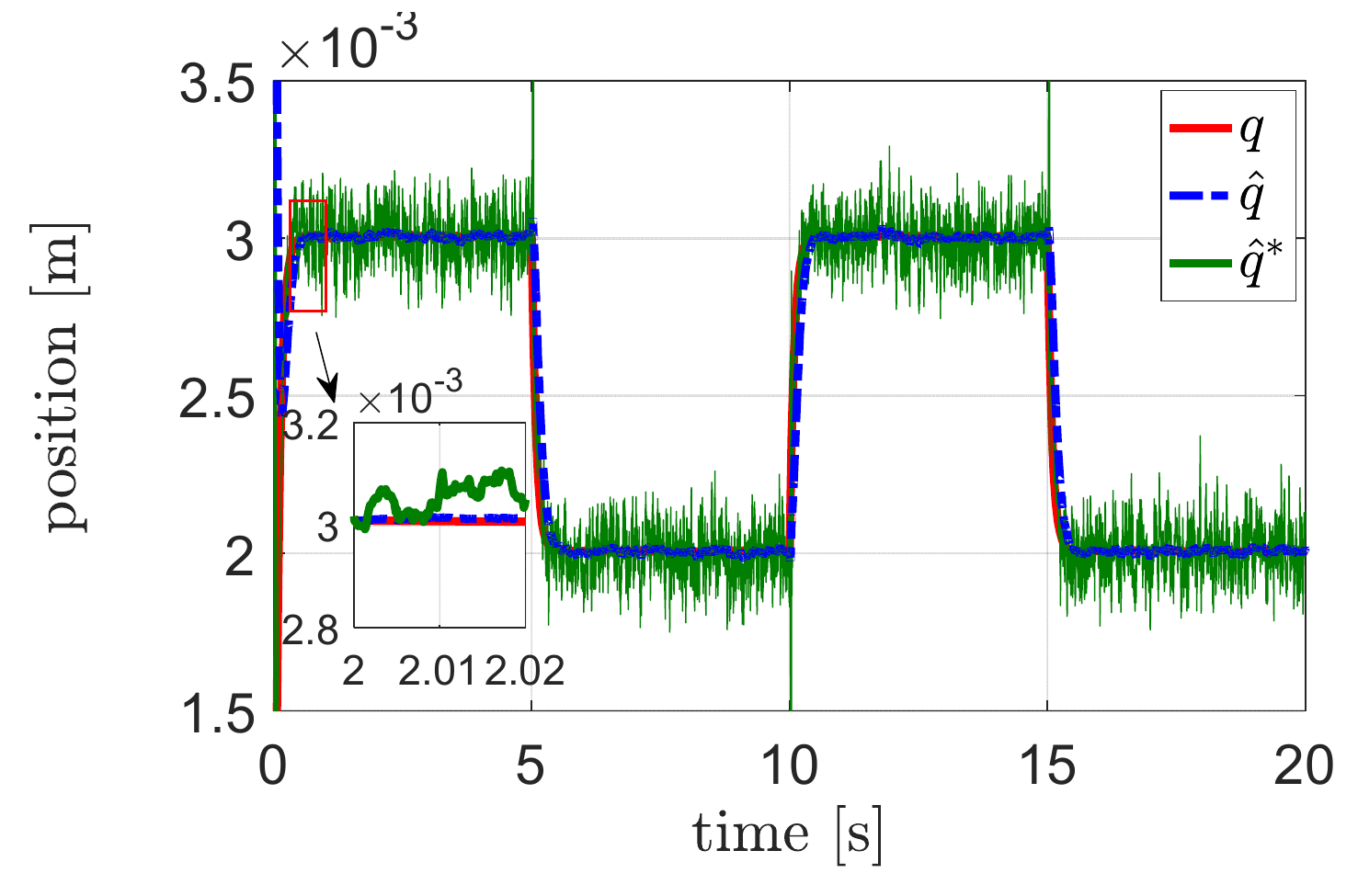}
\includegraphics[width=4.2cm,height=3cm]{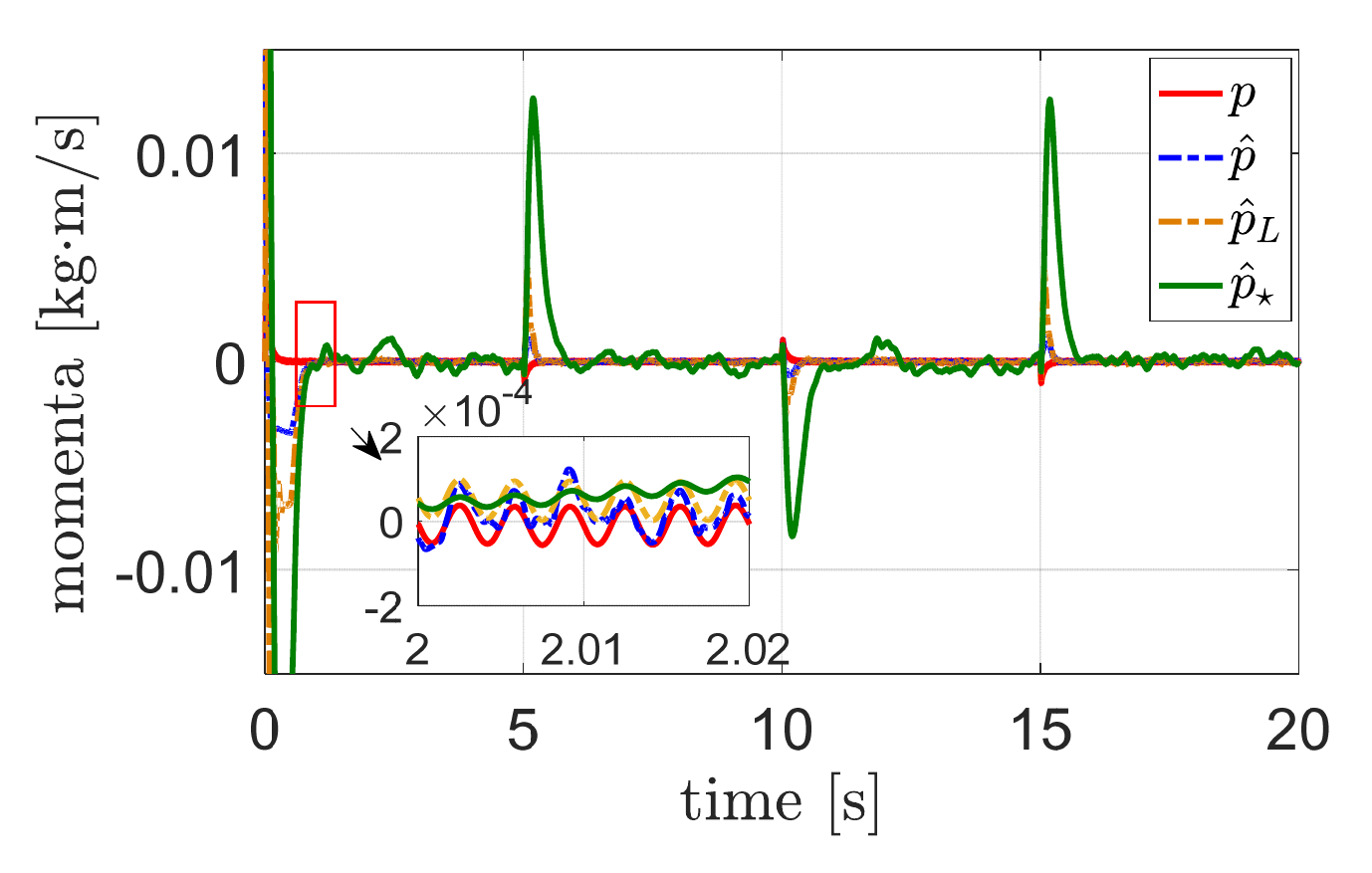}
\includegraphics[width=4.2cm,height=3cm]{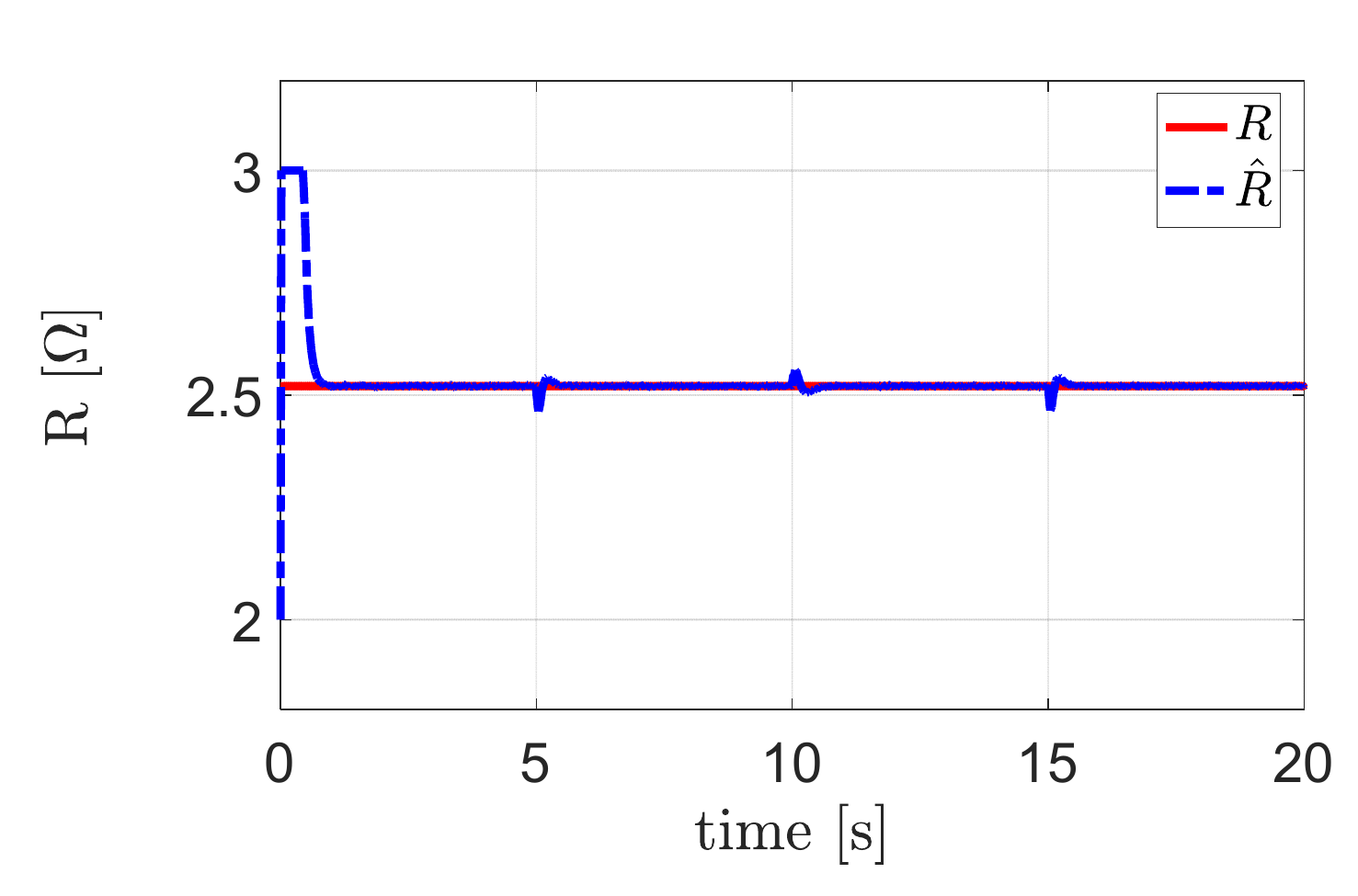}
    \caption{State and parameter estimations (simulation)}
    \label{fig:state}
\end{figure}

\subsection{Performance of the sensorless controller}
\label{subsec62}
In this section, we test the current feedback IDA-PBC law \eqref{uc} with the same parameters as those in Subsection \ref{subsec61}. We observe in Fig. \ref{fig:state-outputfeedback} that the position has a significant regulation error in the first second, which is due to the initial inaccurate estimation of $R$. However, the remaining transients are very satisfactory and almost identical to the state-feedback IDA-PBC.

\begin{figure}[]
    \centering
\includegraphics[width=4.2cm,height=3cm]{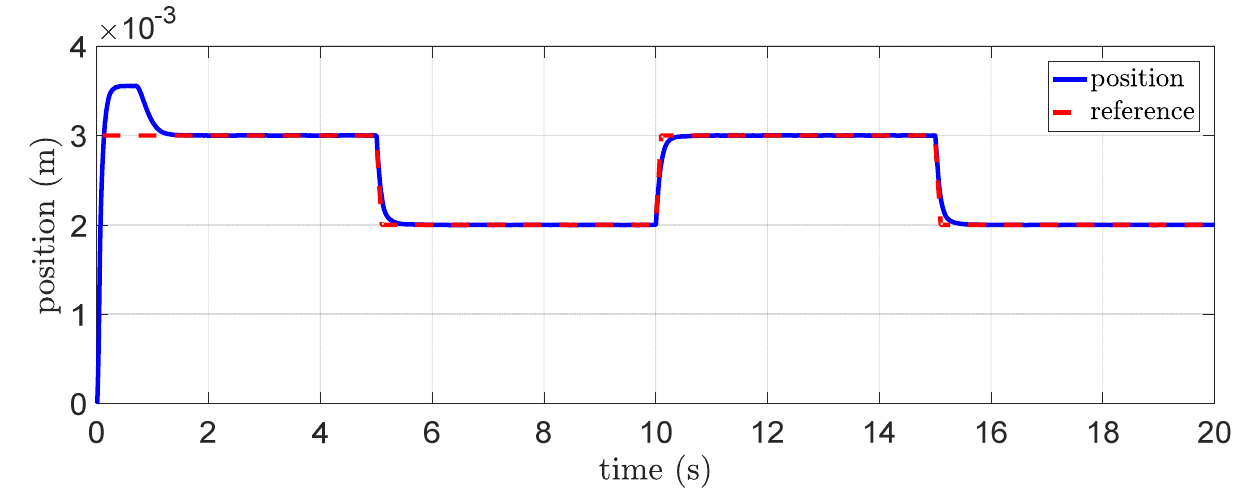}
\includegraphics[width=4.2cm,height=3cm]{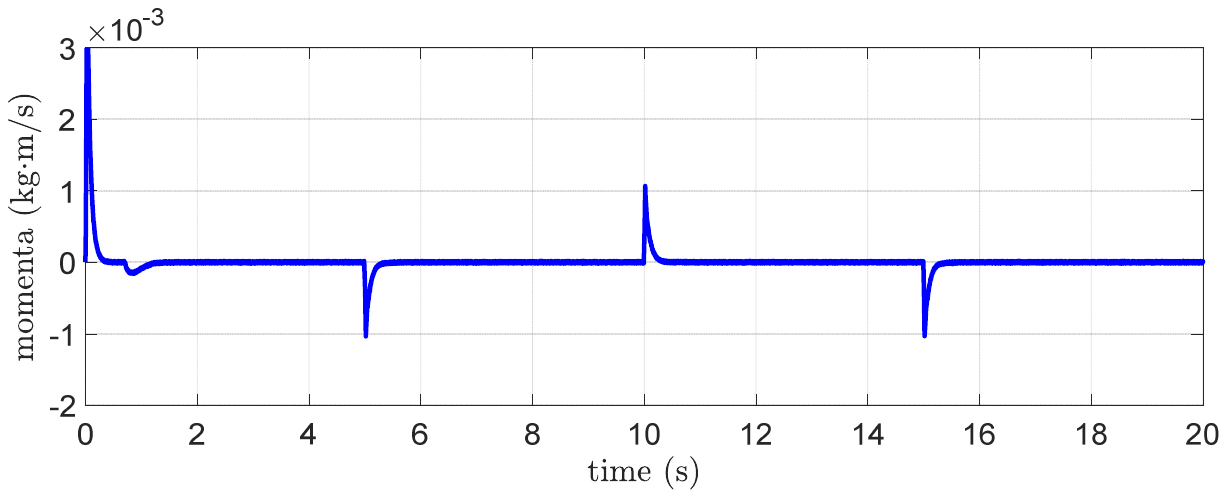}
    \caption{States with sensorless control \eqref{uc}}
    \label{fig:state-outputfeedback}
\end{figure}
\subsection{Experiments}
\label{subsec63}
Some experiments have been conducted on the experimental set-up of the 1-DOF MagLev system shown in Fig. \ref{fig:set-up}, which is located at the laboratory at D\'epartment Automatique, CentraleSup\'elec. The proposed adaptive observer was tested in closed-loop with the following well-tuned backstepping+integral controller
$$
\begin{aligned}
u_0 & = R(c-q)\sqrt{|\Upsilon|}\text{sign}(\Upsilon)  -  K_i \int_{0}^{t}(q - q_{\star} ) d\tau \\
\Upsilon(q,p) & = {2\over k} \big( mg - \gamma_1(p-p_{\star} ) - \gamma_2 m(q- q_{\star})\big),
\end{aligned}
$$
with $K_i=1,\;\gamma_1=340$ and $\gamma_2=3$. The parameters in the observer are taken as $A_0=1.5, \varepsilon=0.03, d=10 \varepsilon,a=10$ and $\gamma=360, \gamma_R=50, \gamma_\lambda=8\times 10^3, \gamma_p=20$.

The responses are shown in Figs. \ref{fig:state_exp}-\ref{fig:output-exp}, where we also give the position estimate $\hat{q}^*$ from the design in \cite{yiscl18}. Unfortunately, the device is only equipped with sensors for position and current. Hence, we can only compare the position estimate with its measured values, as well as the flux linkage estimate with its desired equilibrium. Again, we verify the accuracy and the robustness of the new observer in the presence of measurement noise. Fig. \ref{fig:freq} gives the position estimates with different probing frequencies. It illustrates that a higher frequency yields a higher accuracy, but at the price of a more jittery response.

\begin{figure}[]
    \centering
\includegraphics[width=8cm,height=3.5cm]{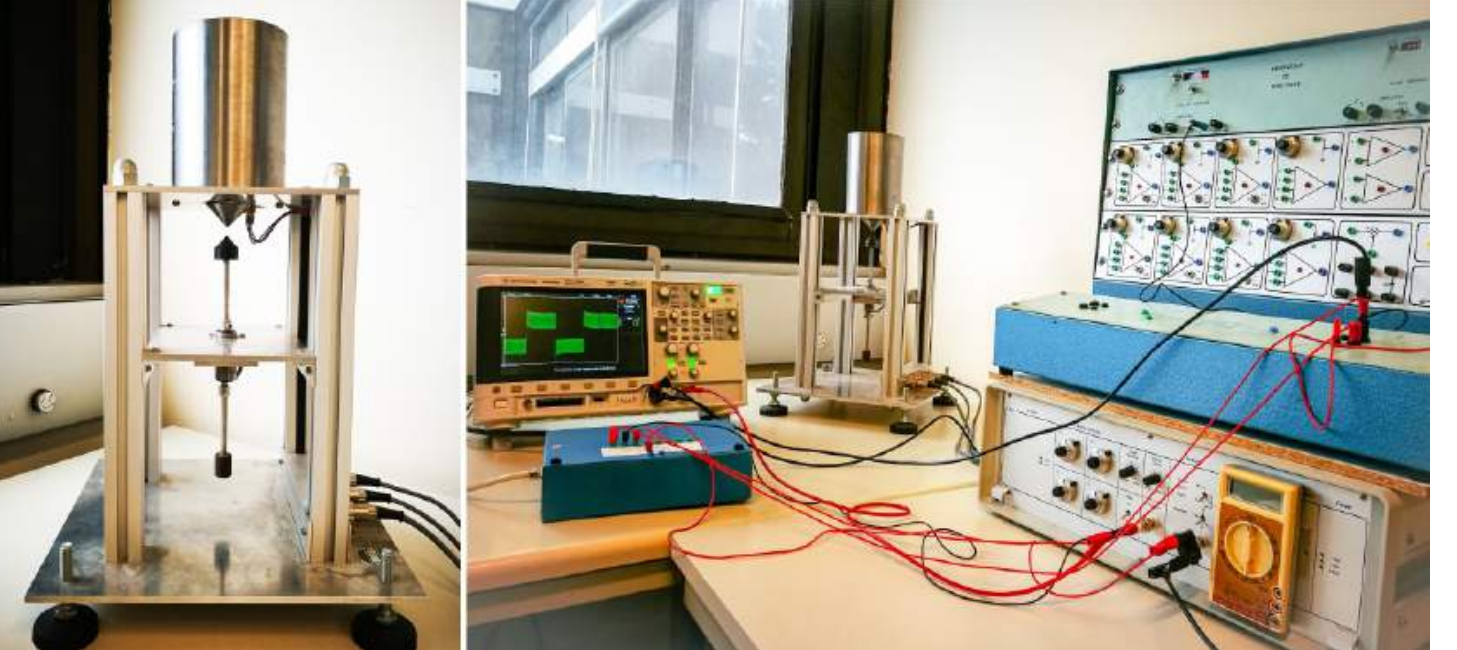}
    \caption{Experimental set-up}
    \label{fig:set-up}
\end{figure}

\begin{figure}[]
    \centering
\includegraphics[width=4.2cm,height=3cm]{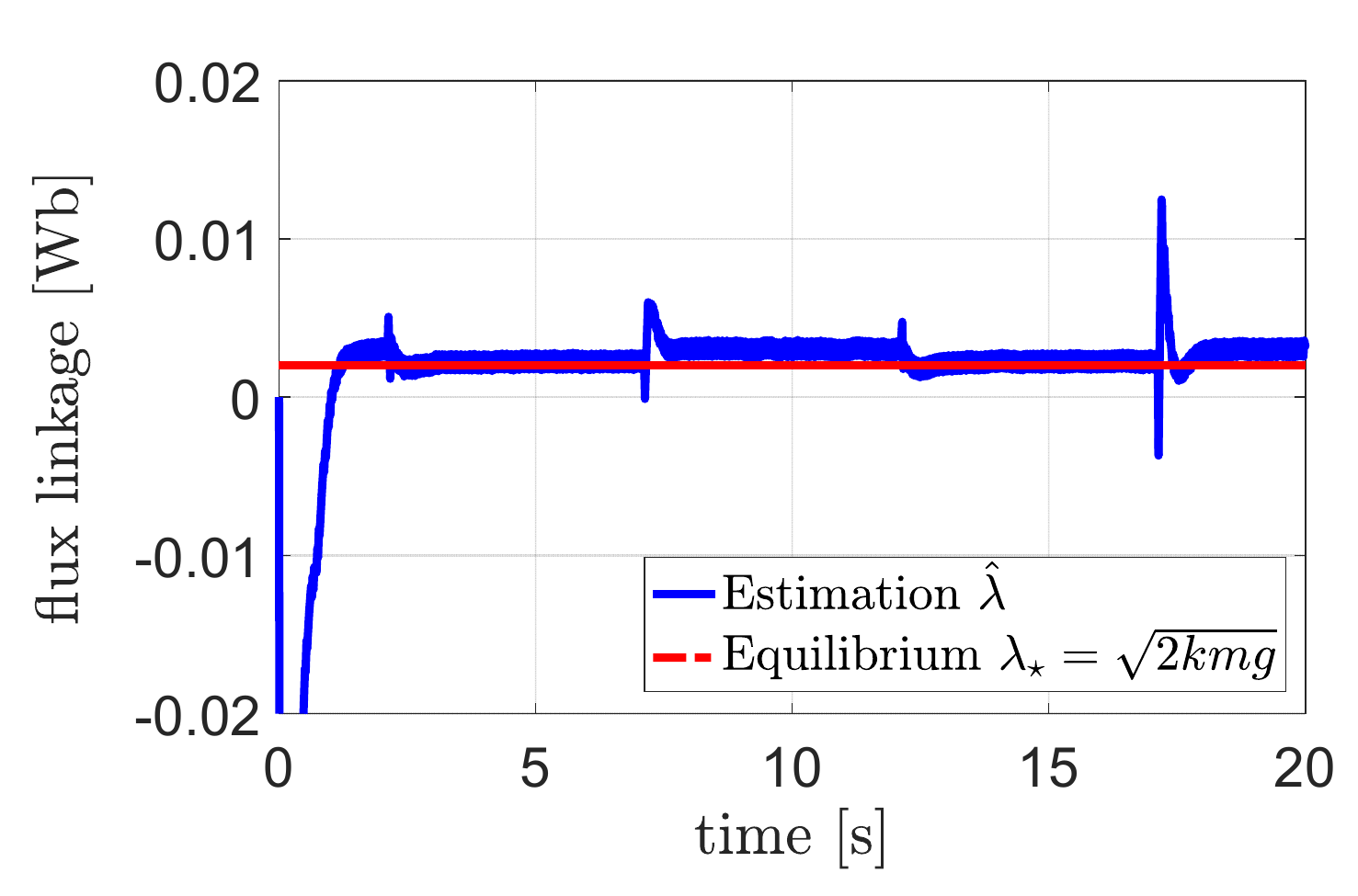}
\includegraphics[width=4.2cm,height=3cm]{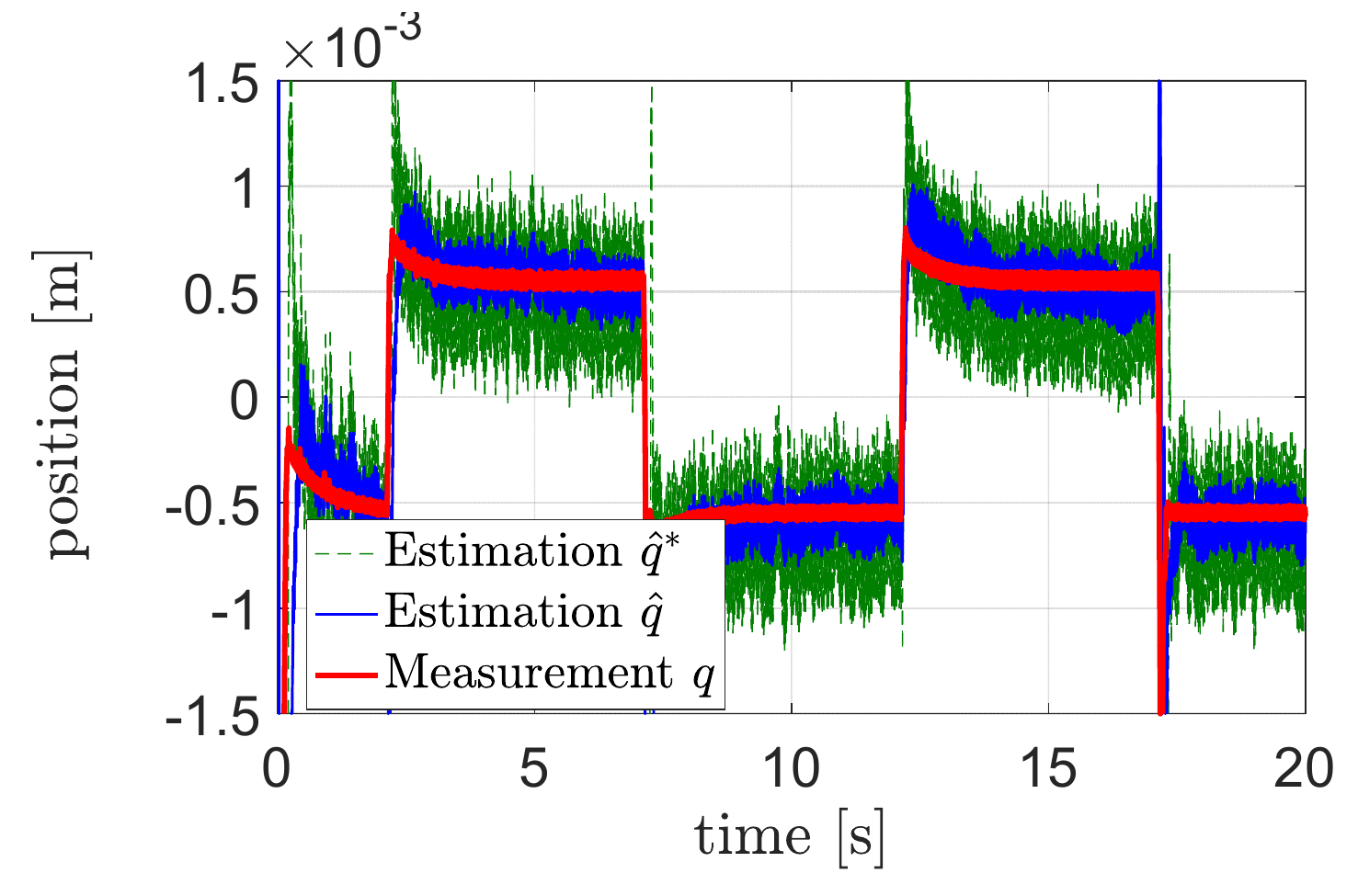}
\includegraphics[width=4.2cm,height=3cm]{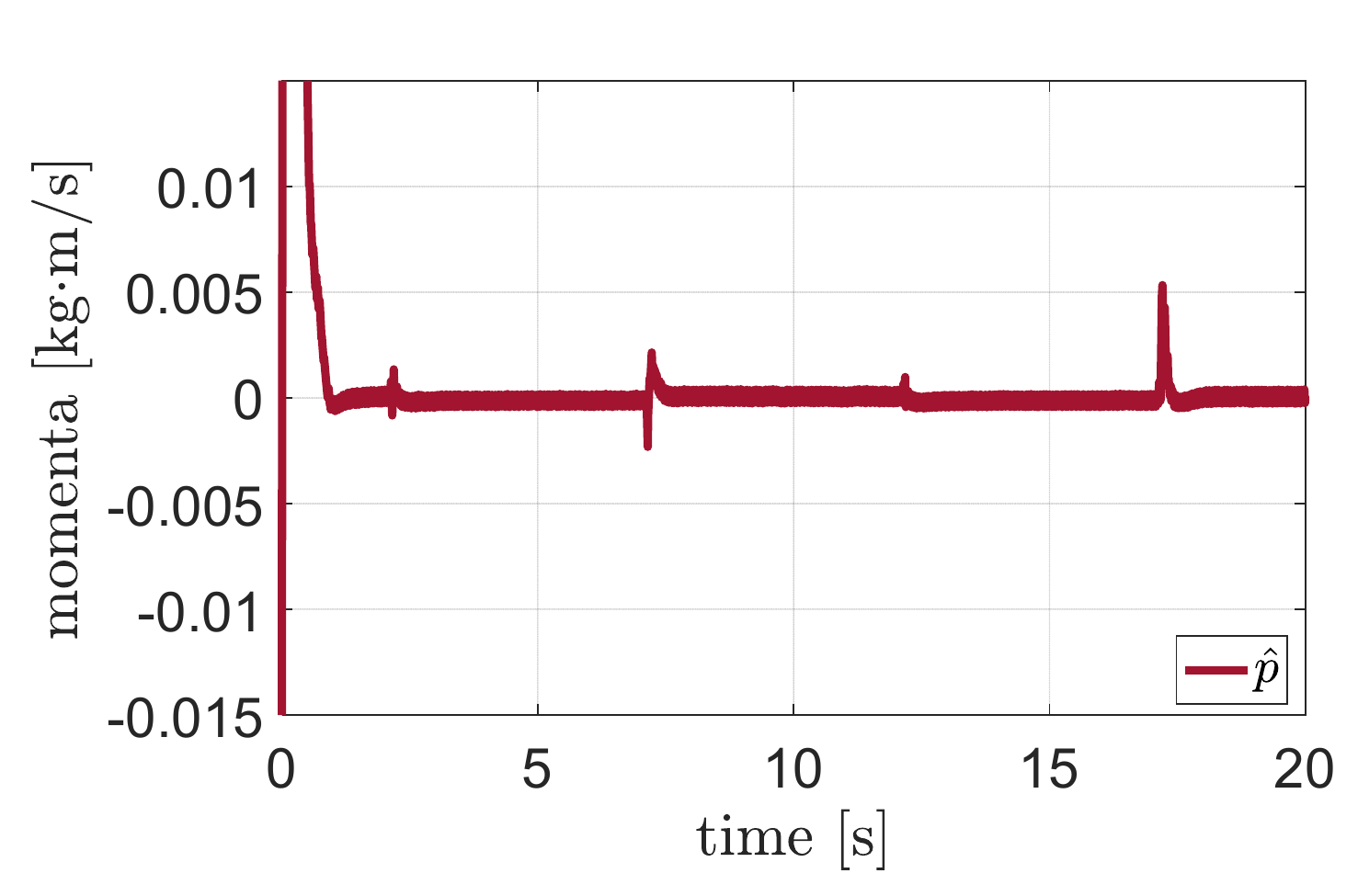}
\includegraphics[width=4.2cm,height=3cm]{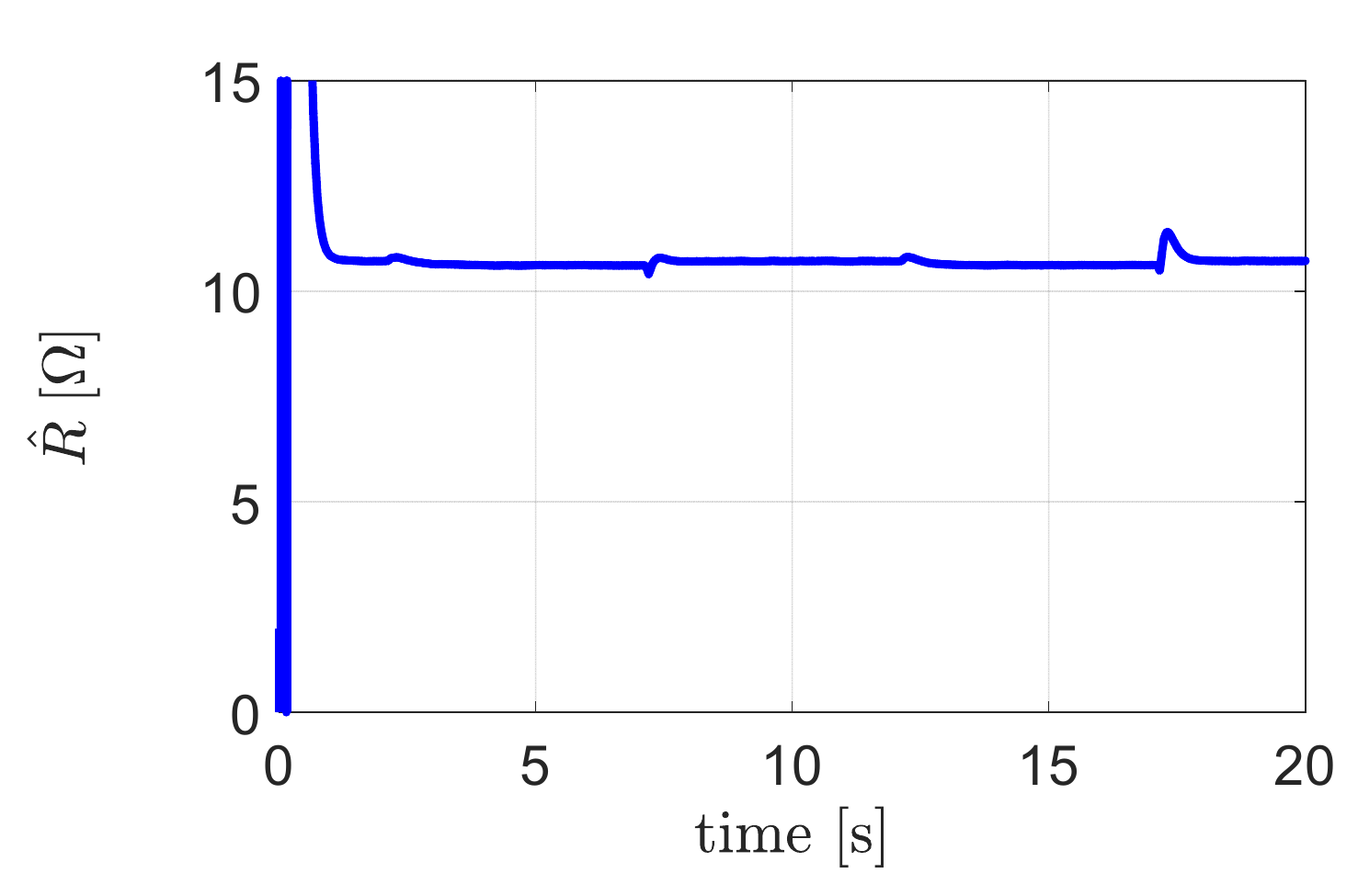}
    \caption{State and parameter estimations (experiment)}
    \label{fig:state_exp}
\end{figure}

\begin{figure}[]
    \centering
\includegraphics[width=4.2cm,height=3cm]{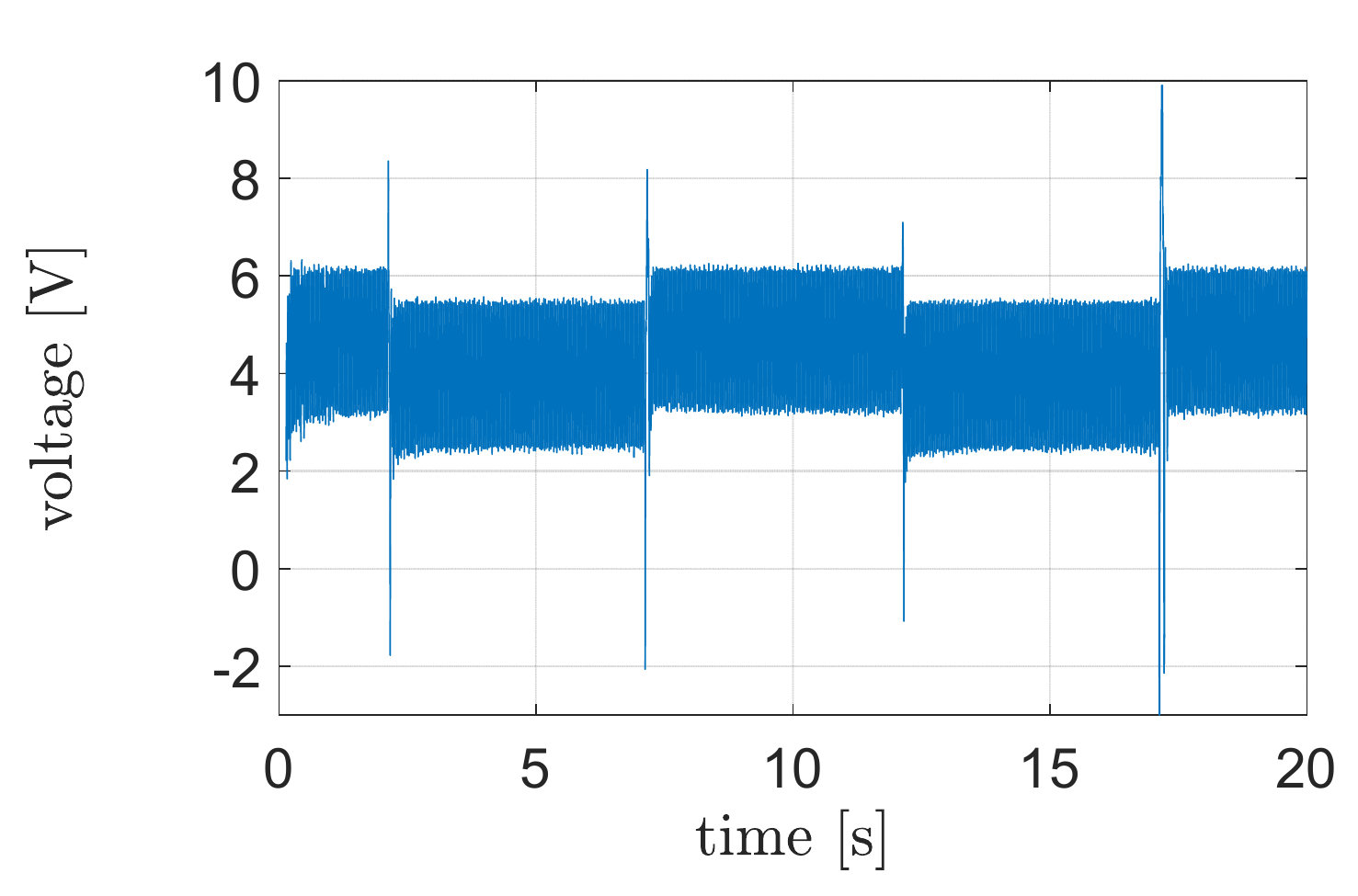}
\includegraphics[width=4.2cm,height=3cm]{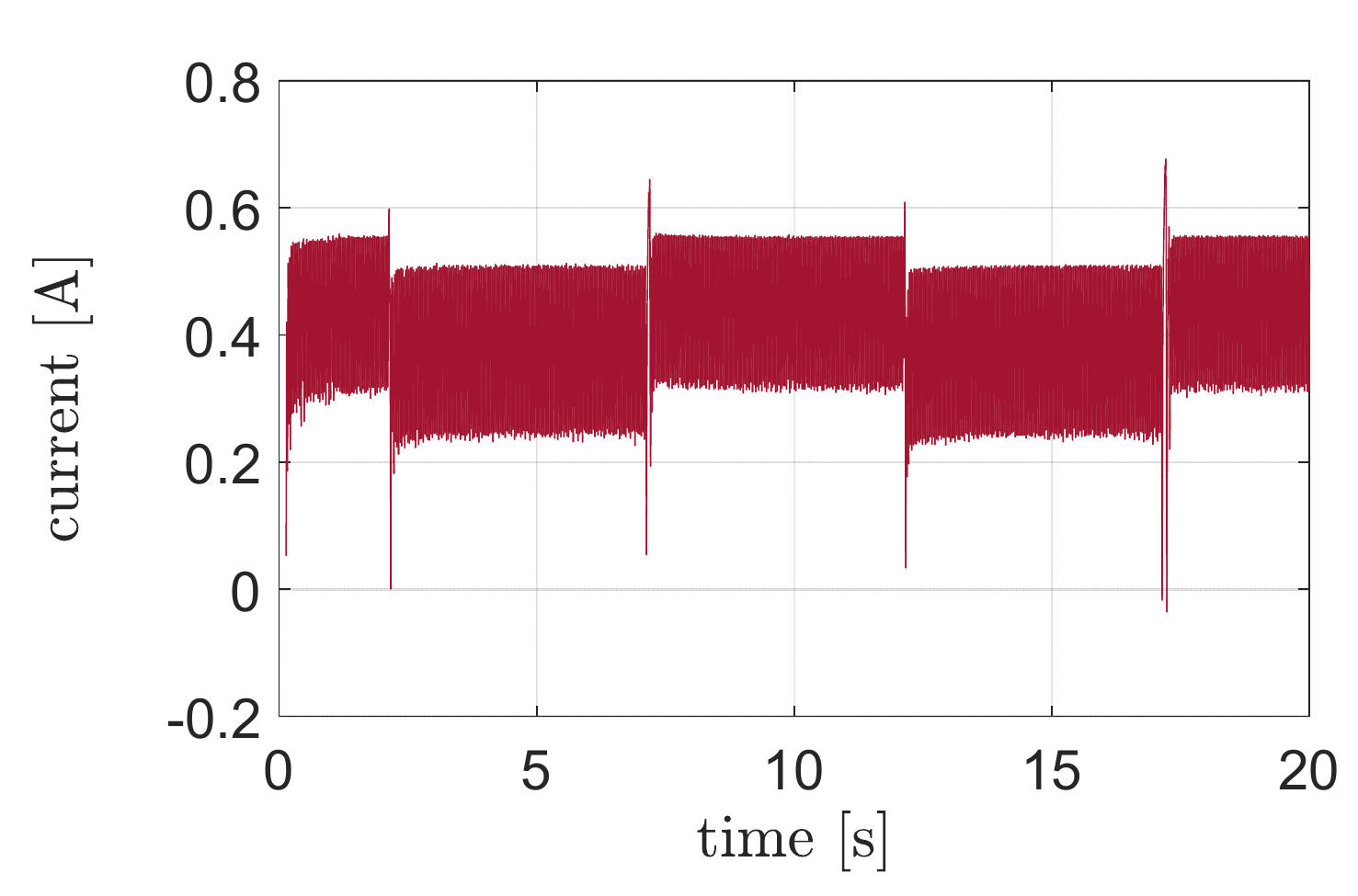}
    \caption{Input and output (experiment)}
    \label{fig:output-exp}
\end{figure}

\begin{figure}[]
    \centering
\includegraphics[width=4.2cm,height=3cm]{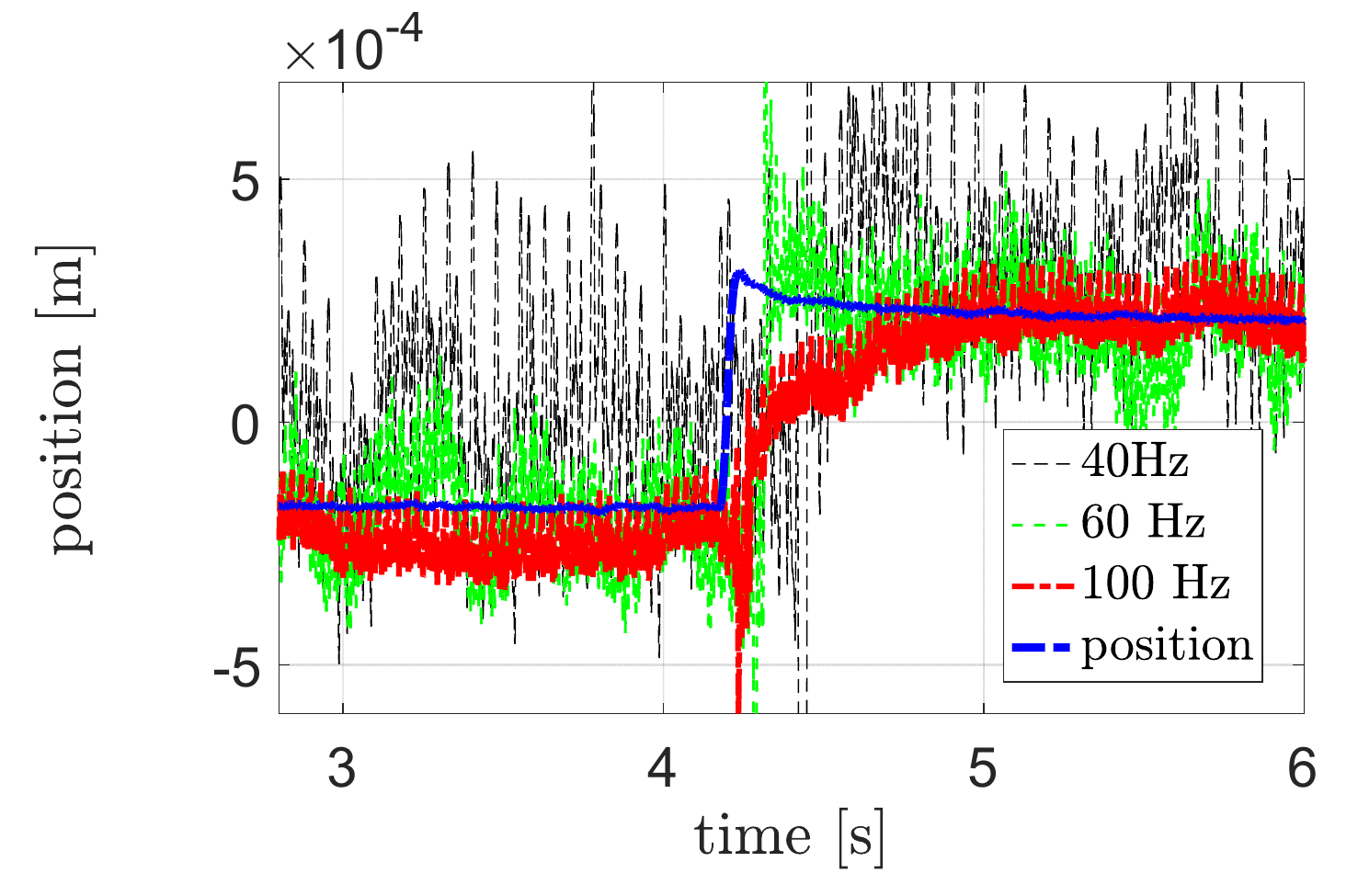}
    \caption{Position estimation for different excitation frequencies}
    \label{fig:freq}
\end{figure}
%
\section{Concluding Remarks}
\label{sec7}
%
In this work we present a novel method for adaptive state observation of a 1-DOF MagLev system, measuring only the coil current and without the knowledge of the electrical resistance. A gradient-descent (like) observer based on DREM is proposed and proven to guarantee---without imposing an excitation assumption---exponential convergence to an (arbitrarily small) residual set. The performance of the observer, and its application in a sensorless controller, has been verified by simulations and experiments.

Some remarks are in order.

\begin{itemize}
\item The 1-DOF MagLev system is a benchmark of electromechanical systems. We are currently investigating the application of the new observer to other electromechanical systems---in particular, electrical motors \cite{yi_draft}.
\item Signal injection is a widely-used technique-oriented method for electromechanical systems. With the notable exception of \cite{combesacc16,COMetal}, no theoretical analysis can be found in the literature. It is challenging to establish the connection between the proposed method and the standard techniques in industry, or provide a theoretical interpretation to the existing technique-oriented methods \cite{yi_draft2}.
\item Although we analyze the effects of the probing frequencies from the theoretical viewpoint, there are many practical considerations that must be taken into account before claiming it to be an operational technique.
\item The simulations and experiments were conducted in the Matlab/Simulink with continuous modules, and as observed the computational efficiency is relatively low. It is of practical interests to give a computationally high-efficient digital implementation of the proposed method.
\end{itemize}


\appendix
\section*{Proof of Lemma \ref{lem1}}
\lab{app1}

First, we will prove that for any smooth signal of the form
  $$
  r(t) = \bar{r}(t) + \varepsilon S(t)r_v(t),
  $$
by fixing $w=2n\varepsilon$ in the WZOH operator with $n \in \mathbb{Z}_+$, we have
\begequ
\lab{cla}
\calz_w[r](t) = \bar{r}(t - {w\over 2}) + \mathcal{O}(\varepsilon^2).
\endequ

According to the definition of the WZOH filter, we have
$$
\begin{aligned}
 \mathcal{Z}_{w}[r](t)
 = & {1 \over w} \int_{0}^{w}r (t-\mu) d\mu,
\end{aligned}
$$
where in the second equality we used the new variable
$
\mu := t -\tau.
$
Thus,
$$
\begin{aligned}
\mathcal{Z}_{w}[r](t)
         =    &  {1 \over w} \int_{0}^{w} \big(\bar{r}(t-\mu) + \varepsilon S(t-\mu)r_v(t-\mu) \big)d\mu\\
         =    & {1 \over w} \int_{0}^{w}\big( \bar{r}(t) - \dot{\bar{r}}(t)\mu + {1\over2}\ddot{\bar{r}}(t)\mu^2 + \delta_2(t,\mu)\mu^3 \big) d\mu  \\
            & + {\varepsilon \over w} \int_{0}^{w} \big( {r}_v(t) - \dot{{r}}_v(t)\mu + \delta_3(t,\mu)\mu^2 \big)
            S(t)d\mu\\
          =  &  \bar{r}(t-{w\over2}) +\varepsilon w \dot{r}_v(t) \mathbb{S}(t)  + \Delta_1(t,w) + \Delta_2(t,w),\\
\end{aligned}
$$
where we have used Taylor expansion---which holds for some mappings $\delta_2$ and $\delta_3$---for the second and third identities, and defined $\mathbb{S}$, as the primitive of $S$, in the fourth identity, and
$$
\begin{aligned}
\Delta_1(t,w) & : = {1 \over w} \int_{0}^{w}\big( {1\over2}\ddot{\bar{r}}(t)\mu^2 + \delta_2(t,\mu)\mu^3 \big) d\mu \\
\Delta_2(t,w) & :=  {\varepsilon \over w} \int_{0}^{w} \big( {r}_v(t) + \delta_3(t,\mu)\mu^2 \big) S(t)d\mu.
\end{aligned}
$$

Invoking $w=2n\varepsilon$, it yields
$$
\varepsilon w \dot{r}_v(t) \mathbb{S}(t)  + \Delta_1(t,w) + \Delta_2(t,w) = \mathcal{O}(\varepsilon^2),
$$
which completes the proof of the claim \eqref{cla}.

Now, it is obvious that,
\begequ
\lab{reg1}
 \mathcal{H}_{{d}}[y](t) = \begin{bmatrix} 1 & S(t-d) \end{bmatrix} \theta(t-d).
\endequ
Applying the property \eqref{cla} of the WZOH filter to the signal $y$ we get
\begequ
\lab{reg2}
\mathcal{Z}_{w}[y](t)  = \begin{bmatrix} 1 & 0 \end{bmatrix} \theta(t-{w\over 2}) + \mathcal{O}(\varepsilon^2).
\endequ
Therefore, selecting $w=2d$, and piling up the two new regressors \eqref{reg1} and \eqref{reg2}, we get
\begin{equation}
\label{eq7}
  \begin{bmatrix} \mathcal{H}_{{d}}[y](t) \\ \mathcal{Z}_{2d}[y](t) \end{bmatrix}
  = \Phi(t)   \theta(t -d)
  + \mathcal{O}(\varepsilon^2),
\end{equation}
where we defined the extended regressor matrix
\begequ
\lab{phi}
\Phi(t) :=   \begin{bmatrix} 1 & S(t -d) \\ 1 & 0 \end{bmatrix}.
\endequ

From \eqref{eq7} we clearly get
\begali{
\mathcal{H}_{{d}}[y](t) -  \mathcal{Z}_{2d}[y](t) & = S(t-d) \theta_2(t-d) + \mathcal{O}(\varepsilon^2),
\label{scareg}
}
completing the proof of the Lemma.

\end{document}